\documentclass[conf,twocolumn]{IEEEtran}
\usepackage{graphicx,cite,calc,dsfont}
\usepackage[defs]{ams}
\usepackage{eqnarray,subeqn,xcolor}
\usepackage{hyperref}
\usepackage{amsthm,amssymb}
\newtheorem{thm}{Theorem}
\newtheorem{lem}{Lemma}
\newtheorem{cor}{Corollary}
\theoremstyle{definition}
\newtheorem{examp}{Example}

\begin{document}

\title{Throughput of One-Hop Wireless Networks with Noisy Feedback Channel}

\author{ Seyed Pooya Shariatpanahi$^1$ , Hamed Shah-Mansouri, Babak Hossein Khalaj$^2$ \\
1: School of Computer Science, Institute for Research in Fundamental Sciences (IPM), Niavaran Sq., Tehran, Iran. \\
2: Department of Electrical Engineering and Advanced Communication Research Institute (ACRI)\\
Sharif University of Technology,
Azadi Ave., Tehran, Iran.\\
(emails: pooya@ipm.ir, hshahmansour@alum.sharif.edu, khalaj@sharif.edu)\\ }

\maketitle
\thispagestyle{empty}
\pagestyle{empty}

\section*{Abstract}
\emph{\let\thefootnote\relax\footnotetext{This research was in part supported by a grant from IPM. \\ Notation: we say $f(n)=o(g(n))$ if we have $|f(n)|<c|g(n)|$ for every positive constant $c$ and large-enough $n$, and $f(n)=\Omega(g(n))$ if $|f(n)|>k|g(n)|$, for some positive constant $k$ and large-enough $n$.  }In this paper, we consider the effect of feedback channel error on the throughput of
one-hop wireless networks under the random connection model. The transmission strategy is based on activating source-destination pairs with strongest direct links. While these activated pairs are identified based on Channel State Information (CSI) at the receive side, the transmit side will be provided with a noisy version of this information via the feedback channel. Such error will degrade network throughput, as we investigate in this paper. Our results show that if the feedback error probability is below a given threshold,
 network can tolerate such error without any significant throughput loss. The threshold value depends on the number of nodes in the network and the channel fading distribution. Such analysis is crucial in design of error correction codes for feedback channel in such networks.
}\\ \\
\emph{\textbf{Index Terms}}---Feedback error, random connection model, throughput, wireless networks.
\section{Introduction}\label{Sec_Introduction}
In order to achieve high throughput in wireless networks diverse techniques have been proposed. For example, \cite{Gupta} and \cite{Fran} consider
the technique of multi-hop routing, equipped with the spatial reuse idea, to arrive at high throughput in large wireless networks. Also, the work \cite{Ayfer}, employs a hierarchical MIMO
scheme to improve network throughput. These works employ the concept of spatial reuse to set up a large number of concurrent transmissions
in the network which heavily relies on the path loss phenomenon in wireless channels.

However, many wireless network scenarios face channel conditions where the dominant effect is the random fading -- and not the large scale path loss -- of the wireless channel (see e.g., \cite{Gowikar} and \cite{Ebrahimi}). The wireless links between nodes in such networks can be modeled by the so called \emph{random connections model}. Since in such networks the spatial reuse idea becomes useless, another approach -- based on the concept of multiuser diversity -- is followed to maximize the throughput. Simply put, in the scheme based on multiuser diversity, the channels with best conditions are activated in order to mitigate the interference effect (\cite{Gowikar}, \cite{Cui}, \cite{Ebrahimi},  \cite{Ebrahimi_2011}, \cite{Pooya_WCL}, and \cite{Pooya_2013}).
In \cite{Gowikar}, a multi-hop scheme is proposed for such networks. In \cite{Cui}, it is discussed that in the absence of spatial reuse opportunity,
using multiple hops does not improve the throughput, and thus, the main research efforts should be focused on single and dual hop strategies. The papers \cite{Ebrahimi}, \cite{Ebrahimi_2011}, \cite{Pooya_WCL}, \cite{Pooya_2013} consider one-hop strategies in such networks.

All the above works, under the random connection model, rely on strategies based on the  multiuser diversity concept, which is itself based on CSI. Since CSI is
available only at the receiver side, there should be a feedback channel through which the transmitters become informed of the transmission strategy. All the above papers
consider that this feedback channel is noiseless, and once the optimum transmission strategy is clear at the receive side (based on the available CSI), the transmit side will also get such information without error. However, such feedback channel is not perfect in practice, leading to noisy feedback received at each transmitter. Analyzing the effect of such imperfection
in the feedback channel on network throughput, is the main issue addressed in this paper.

In this paper, we assume a one-hop wireless network under the random connection model, where channel gains are assumed to be independent and identically distributed (i.i.d.) random variables. There exist $n$ source-destination pairs in our network which communicate in a shared wireless medium. Thus, communication from each source towards
its corresponding destination introduces interference to other pairs. The strategy considered here is to activate the pairs with best direct links, in order to alleviate interference. We investigate
the effect of feedback error on this scheme and analyze the throughput scaling.

It should be noted that considering the effect of noisy feedback channel on operation of wireless systems is a broadly investigated field \cite{Love}.
As the most relevant papers to ours, we can mention \cite{MIMO_1}, \cite{MUD_1}, \cite{MUD_2} and \cite{MUD_3}, which investigate the effect
of feedback error on the schemes based on multiuser diversity in the downlink of cellular networks. It should be noted that -- in contrast to their downlink traffic scenario -- we assume a peer to peer (i.e. unicast) traffic scenario in the network which makes our results totally different from theirs.
In summary, our paper is the first work considering the effect of feedback error on performance of one-hop schemes in networks under random connection model.

The paper is structured as follows. In Section \ref{Sec_Model}, we describe the network model. In Section \ref{Sec_Main}, the main result of
our paper is presented in the form of a theorem. In Section \ref{Sec_Tolerance}, we use the main result to analyze the tolerance of network throughput against feedback error.
Finally, Section \ref{Sec_Conclusion} concludes the paper.

\section{Network Model}\label{Sec_Model}
In this paper, we consider a wireless network consisting of $n$ source-destination pairs located in a shared wireless medium. The $i$th pair consists of the source node $S_i$ and the destination node $D_i$, where, $i=1,\dots,n$. Source $S_i$ aims to transmit its message to the destination $D_i$, and the destination $D_i$ is only interested in decoding the message sent by $S_i$.

\subsection{Wireless Channel Model}
 When the source node $S_i$ transmits its signal, in addition to its corresponding destination (i.e. $D_i$), other destinations (i.e. $D_j$, $j=1,\dots,n$, $j\neq i$) also hear the transmission inevitably through the so-called \emph{broadcast} phenomenon. The signal received by each destination is the \emph{superposition} of the signals transmitted by all sources. Accordingly, we have two kinds of wireless links in the network, namely \emph{direct} and \emph{cross} links. Direct links are those from each source to its corresponding destination, while cross links are from each source to all the $n-1$ non-desired destinations. We model channel power gain of the wireless link from $S_i$ to $D_j$ by the random variable $\gamma_{i,j}$. Based on such model, direct links constitute the set $\{\gamma_{i,i}, i=1,\dots,n\}$, while the cross links constitute the set $\{\gamma_{i,j}, i,j=1,\dots,n, i\neq j\}$. Under the so called \emph{random connection model}, the random variables $\{\gamma_{i,j}, i,j=1,\dots,n\}$ are assumed to be i.i.d. with the common probability distribution function (p.d.f.) $f(\gamma)$, and the corresponding cumulative distribution function (c.d.f.) $F(\gamma)$. We define $\mu=\mathbb{E}\{\gamma\}$, where $\mathbb{E}\{.\}$ stands for the expectation operator. The wireless channel dynamic is assumed to be quasi-static, where during each time slot, the channel power gains are assumed to remain fixed. However, channel gains are changed in consequent time slots, independent of other time slots.

\subsection{Network Operation}
Each source transmits its signal to the corresponding destination in a \emph{single hop} and under an \emph{on-off} transmission strategy. In other words, at each time slot, a subset of sources  -- i.e. $\mathbb{S} \subset \{S_1,...,S_n\}$ -- are set to be \emph{on} and transmit with unit power, and other sources are set to be \emph{off} and remain silent. When the transmission shots by active sources at each time slot is concluded, each destination tries to decode its desired data. We assume single-user decoding at all destination nodes, and accordingly, the successful reception condition at the receiver $i$ is translated to the following Signal to Interference and Noise Ratio ($SINR$) satisfaction constraint:
\begin{equation}\label{Eq_Model_SINR_Definition}
	SINR_{i}\triangleq \frac{\gamma_{i,i}}{N_0+\sum_{S_k \in \mathbb{S}, k\neq i}{\gamma_{k,i}}} \geq \beta,
\end{equation}
where $N_0$ is the power of Additive White Gaussian Noise (AWGN) at the receivers, $\beta$ is a constant threshold, and we have assumed unit transmission power for each source. So, each source-destination pair satisfying the $SINR$ constraint can establish a constant communication rate of $\log(1+\beta)$ at that time slot. In such a setting, we have implicitly assumed that each time slot is long enough for coding and decoding procedures needed to achieve arbitrarily small decoding error probability. Network \emph{throughput} (at that time slot) is defined as the number of those receivers who successfully decode their own message, i.e. those who satisfy the $SINR$ constraint. Such operation scenario is assumed to occur at all time slots.

\subsection{Channel State Information and Feedback Error}
Network throughput depends on the source activation strategy we design for choosing the set of active nodes (i.e. $\mathbb{S}$). This strategy relies on the information we have about the channel power gains (i.e. $\gamma_{i,j}, i,j=1,\dots,n$) in the form of CSI. In this paper, we consider an strategy that only relies on the information of direct links (i.e. $\gamma_{i,i}, i=1,\dots,n$), and ignores the information of cross links (i.e. $\gamma_{i,j}, i,j=1,\dots,n, i\neq j$). In order to further clarify our presentation, we assume that each destination has access to the information of all direct links\footnote{In fact, for the strategy used in this paper, it is enough for each destination to just know its direct link power with the corresponding source, which can be obtained via a training phase at the start of transmission at each time slot \cite{Pooya_2013}. It should be noted that the effect of such training overhead is beyond the scope of our paper.}.
\begin{figure}
\begin{center}
\includegraphics[width=0.51\textwidth]{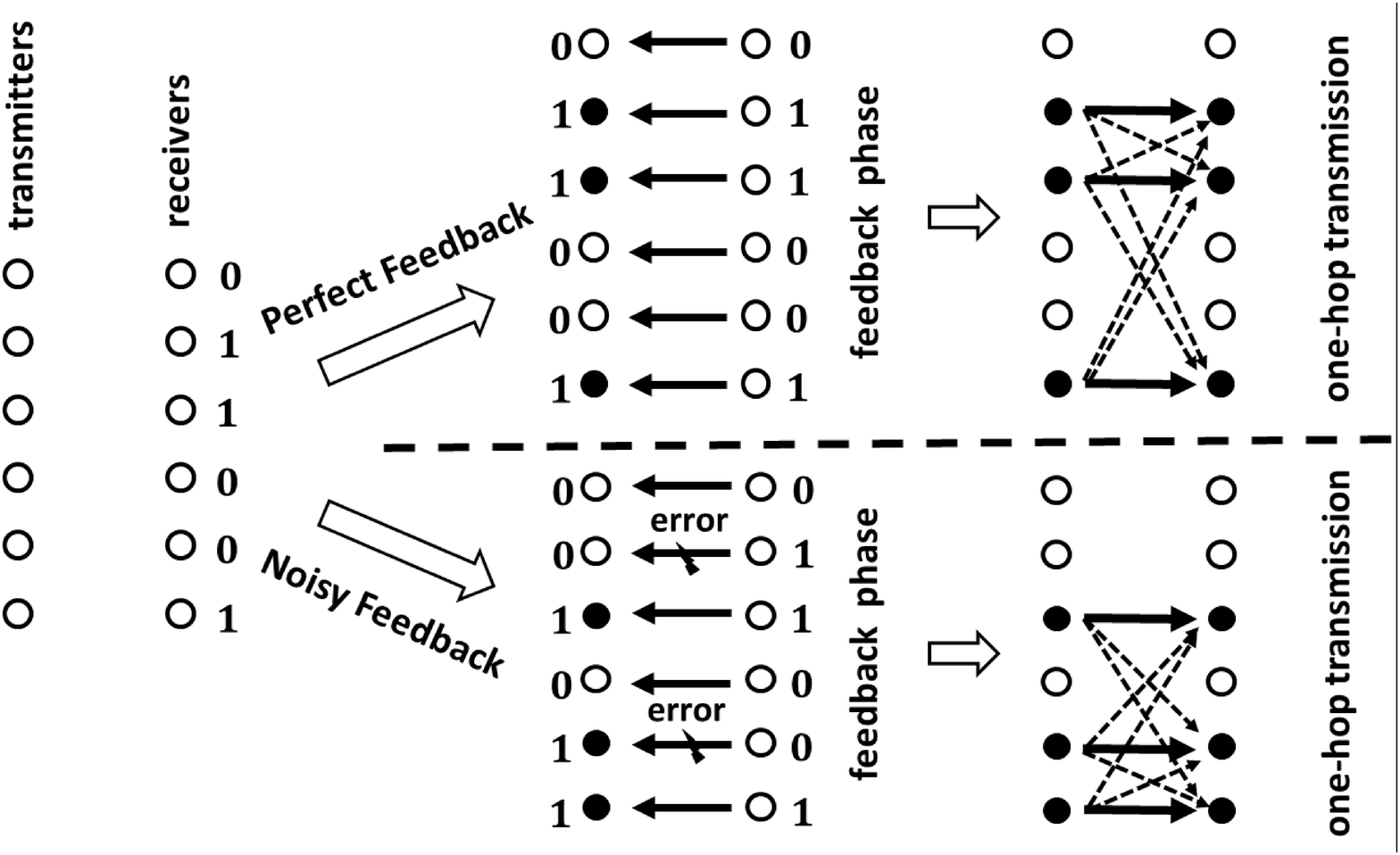}
\end{center}
\caption{Network operation with perfect and noisy feedback channel.\label{Fig1}}
\end{figure}

Based on this CSI, each destination determines if its corresponding source should be active at that time slot, or should remain silent. Then, this binary decision is sent back in form of a one-bit feedback to its corresponding source, via a feedback channel. In this paper, we assume that the feedback channel is noisy, and the one-bit feedback is received at the source with the probability of flip error\footnote{By flip error we mean that a \emph{being ``on''} decision will be understood mistakenly as a \emph{being ``off''} decision, and vice versa.} $\zeta$. The error probability for all feedback channels is the same and the error occurs independently on different channels. Also, to consider the most general setting, we assume that $\zeta$ depends on the number of nodes $n$ as denoted by $\zeta_n$.

Fig. \ref{Fig1} compares the network operation in the cases of perfect and noisy feedback channel. In the upper case in Fig. \ref{Fig1}, the feedback channel is perfect. Therefore, the transmission
strategy determined at the receivers (based on CSI), is sent back to the transmitters perfectly, and subsequently they become active correctly. In contrast, the lower case in Fig. \ref{Fig1} exhibits the noisy feedback scenario, where two feedback errors have occurred in this example.

\section{Achievable Throughput for Networks with Feedback Error}\label{Sec_Main}

We consider the following source activation strategy in this paper. First, we sort the channel power of direct links (i.e. $\gamma_{i,i}, i=1,\dots,n$) to determine their order statistics (i.e. $\gamma_{(i),(i)}, i=1,\dots,n$). Thus, we will have the following:
\begin{equation}\label{Eq_Main_Channel_Sorting}
	\gamma_{(1),(1)} \leq \gamma_{(2),(2)} \leq  \dots \leq \gamma_{(n),(n)}.
\end{equation}
Since each source-destination pair is in a one-to-one correspondence with each direct link, we can accordingly sort the source-destination pairs (i.e. $S_i-D_i, i=1,\dots,n$) into their sorted version (i.e. $S_{(i)}-D_{(i)}, i=1,\dots,n$). Now we have an ordering among the source-destination pairs based on the strength of direct links:
\begin{equation}\label{Eq_Main_Pair_Sorting}
	S_{(1)}-D_{(1)} \leq S_{(2)}-D_{(2)} \leq  \dots \leq S_{(n)}-D_{(n)}.
\end{equation}
The proposed activation strategy activates the maximum number of strongest source-destination pairs. In other words we set:
\begin{equation}\label{Eq_Main_Activation_Set}
	\mathbb{S}= \{ S_{(i)} |  i=n-m+1, \dots, n\},
\end{equation}
where $m$ is chosen as large as possible to maximize the throughput. The next theorem characterizes the achievable throughput of the network operated with this activation strategy, while the feedback channels suffer from imperfection.
\begin{thm}\label{Theorem_Main}
Assume a one-hop network with the following parameters:
\begin{itemize}
\item
$n$ : the number of source-destination pairs
\item
$\zeta_n$ : feedback channel error probability ($\zeta_n \stackrel{n \rightarrow \infty}\rightarrow 0$)
\item
$\beta$ : $SINR$ threshold for successful reception
\item
$\mu$ : channel power average
\item
$F(x)$ : c.d.f. of the channel power
\end{itemize}
Assume all $m>m_0$ satisfy the following inequality:
\begin{equation}\label{Eq_Main_Theorem_Required_Inequality}
(1-\delta_1)F^{-1}\left(1-\frac{m}{n}\right) > \beta K \mu \left(m+(1+\delta_2)\zeta_n(n-m)\right),
\end{equation}
where $m_0$ and $K>1$ are constants, and $\delta_1,\delta_2>0$ are arbitrarily small constants. Subsequently, the network throughput is of order $T=\Omega(m)$.
\end{thm}
\begin{proof}
Let us first present the main ideas behind the theorem and then provide the rigorous proof. First, suppose that the feedback channel is perfect. The strategy for activating source nodes is based on activation of sources with strongest direct links. Suppose we activate $m$ of them (the set in (\ref{Eq_Main_Activation_Set})). Then, the power of direct links of activated source-destination pairs will be the set $\{\gamma_{(i),(i)}, i=n-m+1,\dots,n\}$. Define $r \triangleq n-m+1$. Then, the weakest direct link among these activated source-destination pairs is $\gamma_{(r),(r)}$. Therefore, if the pair $S_{(r)}-D_{(r)}$ is successful in communicating message, then all the other $m-1$ activated pairs will be successful too, and accordingly, we will have throughput of order $m$. Therefore, it is sufficient to focus our analysis on this weakest direct link.

We should select $m$ as large as possible to maximize the throughput. However, increasing $m$ has two coexisting effects. First, $\gamma_{(r),(r)}$ is a decreasing function of $m$, and thus enlarging $m$ will result in weakening $\gamma_{(r),(r)}$. Second, the interference is proportional to the number of activated nodes $m$, thus enlarging $m$ will increase interference level. Therefore, to ensure that the weakest pair satisfies the $SINR$ constraint, we should put a limit on the value of $m$, beyond which interference power will dominate the direct channel power. Such optimum value for $m$ characterizes the target throughput.

Now, suppose that we have a probability of feedback error equal to $\zeta_n$. Consequently, although the chosen activation vector includes pairs with strongest direct links (as in (\ref{Eq_Main_Activation_Set})), the elements of \emph{noisy} activation vector at the transmitter side are flipped with the probability $\zeta_n$. This imperfection introduces two kinds of errors. The first one includes those sources, which according to the original activation vector should have been active, but in the noisy activation vector have become \emph{mistakenly} silent (the second pair from the top in Fig \ref{Fig1}). We call this error, the error of first kind. The second error involves those sources which in the original activation vector were silent, and due to the feedback error have become \emph{mistakenly} active (the fifth pair from the top in Fig. \ref{Fig1}). We call this error, the error of second kind. According to Lemma \ref{Lemma_Error_One}, the error of first kind does not change the order of scaling of the throughput, and thus, can be ignored in the analysis.
\begin{lem}\label{Lemma_Error_One}
Define $E_1$ to be the number of first kind errors. Then, $E_1$ is so small that the presence of the error of first kind will not change the order of throughput, and can be ignored in the scaling analysis.
\end{lem}
Also, according to Lemma \ref{Lemma_Error_Two}, if the number of active sources in the original activation vector is $m$, then the second kind error will not occur more than $\zeta_n (n-m)$ times, with high probability.
\begin{lem}\label{Lemma_Error_Two}
Define $E_2$ to be the number of second kind errors. Then, we will have:
\begin{eqnarray}\label{Eq_Main_Lemma_Error_Two}
\lim_{n\rightarrow \infty} \Pr\{ E_2>(1+\delta_3)\zeta_n(n-m)\} =0 \\ \nonumber
\lim_{n\rightarrow \infty} \Pr\{ E_2<(1-\delta_3)\zeta_n(n-m)\} =0,
\end{eqnarray}
where $\delta_3>0$ is an arbitrarily small constant.
\end{lem}
By ignoring the first kind error (based on Lemma \ref{Lemma_Error_One}), the effect of feedback error can be translated into additional interference generated by the second kind error where this extra interference is proportional to the number of sources experiencing the second kind error, i.e. $\zeta_n (n-m)$.

On the other hand, from a result in the intermediate order statistics context (stated in Lemma \ref{Lemma_Falk} in appendix) it can be seen that $\gamma_{(r),(r)}$ is of order $F^{-1}\left(1-\frac{m}{n}\right)$. Now we can interpret the Left Hand Side (LHS) and Right Hand Side (RHS) of (\ref{Eq_Main_Theorem_Required_Inequality}). LHS(5) is, roughly speaking, the order of the weakest direct link power in the activated set (i.e. $\gamma_{(r),(r)}$). RHS(5) consist of two interference terms. The first one is the interference present even in the case of perfect feedback channel -- which is of order $m$. The second one is the interference imposed by the feedback error of the second kind -- which is of order $\zeta_n(n-m)$. Thus, (\ref{Eq_Main_Theorem_Required_Inequality}) ensures that $m$ grows slowly-enough (as a function of $n$) so that the SINR constraint at the weakest direct link is satisfied. Accordingly, all the $m$ source-destination pairs will be successful, and the throughput of order $m$ can be achieved.

After presenting the main ideas behind the theorem, we provide the rigorous proof. Suppose that the number of activated sources according to the original activation vector is $m$. Define:
\begin{eqnarray}\label{Eq_Main_Rigorous_Proof_Definitions}
l&\triangleq&m+(1+\delta_2)\zeta_n(n-m) \\ \nonumber
\phi&\triangleq&K\mu l,
\end{eqnarray}
where $K>1$. Then, we will have Lemma \ref{Lemma_Direct_Power} for the power of direct links of these $m$ strongest source-destination pairs:
\begin{lem}\label{Lemma_Direct_Power}
If $m$ satisfies inequality (\ref{Eq_Main_Theorem_Required_Inequality}), we will have
\begin{equation}\label{Eq_Main_Lemma_Direct_Power}
\lim_{n\rightarrow \infty}\Pr\{\gamma_{(r),(r)} >  \beta \phi \cap \dots \cap \gamma_{(n),(n)}>  \beta  \phi\} =1.
\end{equation}
\end{lem}
Also, if $I_i$ is defined to be the interference imposed to the pair $S_{(i)}-D_{(i)}$, from our above discussion we
know that it consists of two kinds of interference. When considered together, the total interference
can be considered as the sum of $l$ i.i.d. random variables with the c.d.f. $F(\gamma)$, where $l$ is defined in (\ref{Eq_Main_Rigorous_Proof_Definitions}). Then, we will have the following lemma:
\begin{lem}\label{Lemma_Interference}
\begin{equation}\label{Eq_Main_Lemma_Interference}
\lim_{n\rightarrow \infty}\Pr\{I_r \leq \phi \cap \dots \cap I_n \leq  \phi\} =1.
\end{equation}
\end{lem}
Lemmas \ref{Lemma_Direct_Power} and \ref{Lemma_Interference}, when considered together, ensure that $SINR$ at all the $m$ strongest source-destination pairs will be greater than $\beta$. It should be noted that, although some of these activated source-destination pairs are mistakenly inactivated due to error of first kind, the order of throughput loss by such fact is less than $m$ (a fact made precise in Lemma \ref{Lemma_Error_One}). Therefore, this effect does not change the number of successful receptions scaling. This concludes the proof (proofs of Lemmas \ref{Lemma_Error_One}, \ref{Lemma_Error_Two}, \ref{Lemma_Direct_Power}, and \ref{Lemma_Interference} are provided in the appendix.).
\end{proof}

\section{Network Feedback Error Tolerance}\label{Sec_Tolerance}
Next, we exploit Theorem \ref{Theorem_Main} to analyze sensitivity of network throughput to the feedback error. The main result of this section is presented in the following Corollary:
\begin{cor}\label{Corr_Tolerance}
Suppose for a network setting with perfect feedback channels we achieve the throughput of order $T(n)$. Then, for the same network setting with the feedback error probability $\zeta_n=o\left(\frac{T(n)}{n}\right)$, throughput of the same order $T(n)$ is achievable.
\end{cor}
\begin{proof}
The intuition behind the corollary is straightforward. If the feedback channels are perfect and the throughput of $T(n)$ is achievable, then
we have $T(n)$ active sources. Therefore, the interference will be of order $T(n)$ in the absence of feedback error. Feedback error introduces the extra interference generated by the second kind error. Roughly speaking, this interference is at most of order $\zeta_n n = o(T(n))$ which is asymptotically smaller than the original interference, and can be ignored compared to that.
Thus, introduction of feedback error will not change the order of total interference power, and consequently, the same number of pairs can be activated in this case.

The rigorous proof is as follows. We assume that in the noise-free case the throughput $T(n)$ is achievable, then we prove that the throughput $m=c_1T(n)$ is achievable in the noisy case, for small-enough constant $c_1$. We have:
\begin{eqnarray}\label{Eq_Tolerance_Corollary_Proof}
LHS(5)&=&(1-\delta_1)F^{-1}\left(1-\frac{m}{n}\right)  \\ \nonumber
&\stackrel{(a)}=&(1-\delta_1)F^{-1}\left(1-\frac{c_1 T(n)}{n}\right) \\ \nonumber
&\stackrel{(b)}>& (1-\delta_1)F^{-1}\left(1-\frac{T(n)}{n}\right) \\ \nonumber
&\stackrel{(c)}>&\beta K \mu T(n) \\ \nonumber
&\stackrel{(d)}>&\beta K \mu (c_1+c_2(1+\delta_2))T(n) \\ \nonumber
&=&\beta K \mu (c_1T(n)+c_2\frac{T(n)}{n}(1+\delta_2)n) \\ \nonumber
&>&\beta K \mu (c_1T(n)+c_2\frac{T(n)}{n}(1+\delta_2)(n-c_1T(n))) \\ \nonumber
&\stackrel{(e)}>&\beta K \mu (m+\zeta_n(1+\delta_2)(n-m)) \\ \nonumber
&=&RHS(5),
\end{eqnarray}
where (a) follows from the fact that we have set $m=c_1T(n)$, (b) follows from the facts that $F^{-1}(x)$ is a strictly increasing function and $c_1<1$, (c) follows from our assumption that the throughput $T(n)$ is achievable in the noise-free case (applying inequality (\ref{Eq_Main_Theorem_Required_Inequality}) with $\zeta_n=0$, and $m=T(n)$), (d) follows from the fact that $c_1,c_2>0$ are arbitrarily small constants, and (e) follows from our assumption that $\zeta_n=o\left(T(n)/n\right)$.

Thus, we have proved that if throughput $T(n)$ with $\zeta_n=0$ is achievable, then throughput of order $T(n)$ with $\zeta_n=o\left(T(n)/n\right)$ is also achievable.
\end{proof}
In summary, this corollary states that the feedback error of order $T(n)/n$ is tolerable without paying any price in terms of throughput scaling. Next examples are applications of this corollary to specific channel distributions.
\begin{examp}
Consider a network with Rayleigh fading wireless channels. Then, the c.d.f. of power of such channel is:
\begin{equation}\label{Eq_Tolerance_Example_Rayleigh_CDF}
F(x)=1-e^{-x/\mu},
\end{equation}
where $\mu$ is the average channel power. First, we show that the throughput of order $\log(n)$ in a network with noiseless feedback channel is achievable. In order to do that, by choosing $c<(1-\delta_1)/\beta K$, for large enough $n$ we will have:
\begin{eqnarray}\label{Eq_Tolerance_Example_Rayleigh_Without_Error_1}
c\frac{\log(n)}{n} &\leq& \frac{1}{n^{\beta K  c / (1-\delta_1)}} \\ \nonumber
&=& e ^ {-\beta K  c \log(n)/(1-\delta_1)},
\end{eqnarray}
which can be rewritten as
\begin{eqnarray}\label{Eq_Tolerance_Example_Rayleigh_Without_Error_2}
1-c\frac{\log(n)}{n} &>& 1 -  e ^ {-\beta K c \log(n)/(1-\delta_1)} \\ \nonumber
&=& F \left( \frac{\beta K \mu c \log(n)}{1-\delta_1} \right).
\end{eqnarray}
Since $F^{-1}(x)$ is a strictly increasing function, we can arrive at
\begin{eqnarray}\label{Eq_Tolerance_Example_Rayleigh_Without_Error_3}
F^{-1}\left(1-\frac{c\log(n)}{n}\right) > \frac{\beta K \mu c \log(n)}{1-\delta_1},
\end{eqnarray}
or equivalently
\begin{eqnarray}\label{Eq_Tolerance_Example_Rayleigh_Without_Error_4}
LHS(5)&=&(1-\delta_1)F^{-1}\left(1-\frac{m}{n}\right) \\ \nonumber
&\stackrel{(a)}=&(1-\delta_1)F^{-1}\left(1-\frac{c\log(n)}{n}\right) \\ \nonumber
&\stackrel{(b)}>& \beta K \mu c \log(n) \\ \nonumber
&=& \beta K \mu m \\ \nonumber
&\stackrel{(c)}=&RHS(5),
\end{eqnarray}
where (a) follows from the fact that we have set $m=c\log(n)$, (b) follows from (\ref{Eq_Tolerance_Example_Rayleigh_Without_Error_3}), and (c) is due to the fact that we have assumed $\zeta_n=0$. Accordingly, based on Theorem \ref{Theorem_Main}, we have proved that in a network with perfect feedback channel (i.e. $\zeta_n=0$), the throughput of order $\log(n)$ (i.e. $m=c\log(n)$, for constant $c$) is achievable.

Therefore, according to corollary \ref{Corr_Tolerance}, in this network with a feedback error probability of order $o\left(\log(n)/n\right)$, throughput of the same order $\log(n)$ is achievable. This fact has an important practical implication. It means, by designing a powerful-enough error correction coding scheme -- which suppresses feedback error probability to order $o\left(\log(n)/n\right)$ -- we can alleviate the harmful effect of feedback error. Moreover, designing more powerful error correction coding schemes will not be of further value.
\end{examp}

\begin{examp}
Consider a network in which the c.d.f. of the underlying channel power has the form
\begin{equation}\label{Eq_Tolerance_Example_Pareto_CDF}
F(x)=1-\frac{1}{(1+x)^\alpha},
\end{equation}
where $\alpha>2$ is the distribution parameter. First, we note that in the case of perfect feedback channel, the throughput of order $n^{1/(1+\alpha)}$ is achievable. In order to clarify this issue, we set $m=c n^{1/(1+\alpha)}$. Then, we will have
\begin{eqnarray}\label{Eq_Tolerance_Example_Pareto_Without_Error_1}
LHS(5)&=&(1-\delta_1)F^{-1}\left(1-\frac{m}{n}\right)\\ \nonumber
&\stackrel{(a)}=& (1-\delta_1)\left(\left(\frac{n}{m}\right)^{1/\alpha} -1 \right) \\ \nonumber
&\stackrel{(b)}=&(1-\delta_1)\left( \frac{1}{c^{1/\alpha}} n^{1/(\alpha+1)} -1 \right) \\ \nonumber
&\stackrel{(c)}>& \beta K \mu c n^{1/(\alpha+1)} \\ \nonumber
&=& \beta K \mu m \\ \nonumber
&\stackrel{(d)}=& RHS(5),
\end{eqnarray}
where (a) comes from (\ref{Eq_Tolerance_Example_Pareto_CDF}), (b) is due to the fact that we have put $m=c n^{1/(1+\alpha)}$, (c) is valid for small-enough constant $c$, and (d) is because we have assumed $\zeta_n=0$. Thus, we have shown that inequality (\ref{Eq_Main_Theorem_Required_Inequality}) is valid for $\zeta_n=0$ and $m=c n^{1/(1+\alpha)}$ (which indicates achievability of throughput of order $n^{1/(1+\alpha)}$ in the error-free network). Then, according to Corollary \ref{Corr_Tolerance}, we know that throughput of the same order $n^{1/(1+\alpha)}$ is achievable even if we have feedback error probability of order $o\left(T(n)/n\right)=o\left(n^{-\alpha/(1+\alpha)}\right)$. Here, we have the same practical guideline for error correction coding design, as in the previous example.
\end{examp}
\section{Conclusions}\label{Sec_Conclusion}
In this paper, we have analyzed the effect of feedback error on performance of one-hop communications in wireless networks.
The channel model assumed was the random connection model, and the transmission strategy was based on activating source-destination pairs with
strongest direct links. Although the feedback error degrades network throughput, we have proved that there is a threshold for
the feedback error probability, below which the network can tolerate the error, in the scaling sense. This threshold is of order $o(\log(n)/n)$ for the Rayleigh fading case,
and of order $o(n^{-\alpha/(\alpha+1)})$ for the Pareto power distribution. These results are of great importance in the design of error correction scheme for the feedback channels.

In practice, in addition to error, feedback information experiences delay in the feedback channel. Analyzing the effect of delay under our network model assumptions and for unicast traffic, is a promising direction for further research.

\section*{Appendix}\label{Sec_Appendix}
\begin{proof}[Proof of Lemma \ref{Lemma_Error_One}]
To prove this lemma we show that (assuming that in the original activation vector, $m$ nodes are active) the number of source-destination pairs encountering the error of first kind is less than $m$ in the scaling sense. Thus, these errors will not harm the scaling of the number of successful pairs.
Let us consider two cases:
\begin{itemize}
\item
$\lim_{n\rightarrow \infty}{\zeta_n m}= cte.\geq 0:$ \\
In this case, the number of pairs experiencing error of the first kind will be a Poisson random variable with
bounded average, and thus can be asymptotically ignored in comparison with $m$.
\item
$\lim_{n\rightarrow \infty}{\zeta_n m} = \infty:$ \\
Define the binary random variable $B_i$ to indicate the event that the $i$th source among the first $m$ strongest pairs experiences
the error of first kind. Then, using a well established probabilistic discussion, we will have:
\begin{eqnarray}\label{Eq_Appendix_Lemma_Error_One_1}
& &\Pr\{E_1>(1+\delta_3)\zeta_n m \} \\ \nonumber
 &=& \Pr\{\sum_{i=1}^{m}{B_i}>(1+\delta_3)\zeta_n m \} \\ \nonumber
&=&\Pr\{e^{s\sum_{i=1}^{m}{B_i}}>e^{s(1+\delta_3)\zeta_n m} \} \\ \nonumber
&\stackrel{(a)}\leq& \frac{\left(\mathbb{E}\{e^{sB_1}\}\right)^m}{e^{s(1+\delta_3)\zeta_n m}} \\ \nonumber
&=& \frac{\left(e^s \zeta_n+ (1-\zeta_n)\right)^m}{e^{s(1+\delta_3)\zeta_n m}} \\ \nonumber
&\stackrel{(b)}<& \frac{e^{\zeta_n m (e^s-1)}}{e^{s(1+\delta_3)\zeta_n m}} \\ \nonumber
&\stackrel{(c)}=&e^{-\zeta_n m \Lambda(\delta_3)} \\ \nonumber
&\rightarrow& 0,
\end{eqnarray}
where (a) is due to Markov's inequality and independence of feedback errors for different pairs, (b) is due to the identity $1+x<e^x$ for $x>0$, and (c) is by putting $s=\log(1+\delta_3)$ and defining $\Lambda(x) \triangleq (1+x)\log(1+x) - x > 0$.
Since $\zeta_n \rightarrow 0$, $E_1$ is asymptotically dominated by $m$.
\end{itemize}
\end{proof}
\begin{proof}[Proof of Lemma \ref{Lemma_Error_Two}]
The proof of this lemma is very similar to the proof of Lemma \ref{Lemma_Error_One}.
\end{proof}
\begin{proof}[Proof of Lemma \ref{Lemma_Direct_Power}]
Before proving this lemma, we need two other lemmas which we present first:
\begin{lem} [Falk, 1989] \label{Lemma_Falk}
Assume that $X_1,X_2,\dots,X_n$ are i.i.d. random variables with the c.d.f. $F(x)$. Define $X_{(1)},X_{(2)}, \dots, X_{(n)}$ to be the order statistics of $X_1,X_2,\dots,X_n$. If $i \rightarrow \infty$ and $i/n \rightarrow 0$ as $n \rightarrow \infty$, then there exist sequences $a_n$ and $b_n>0$ such that
\begin{equation}\label{Eq_Appendix_Lemma_Falk_1}
	\frac{X_{(n-i+1)}-a_n}{b_n} \Rightarrow N(0,1),
\end{equation}
where $\Rightarrow$ denotes convergence in distribution, and $N(0,1)$ is the Normal distribution with zero mean and unit variance. Furthermore, one choice for $a_n$ and $b_n$ is:
\begin{eqnarray}\label{Eq_Appendix_Lemma_Falk_2}
	a_n=F^{-1}\left(1-\frac{i}{n}\right), \;\;\;	b_n=\frac{\sqrt{i}}{nf(a_n)}.
\end{eqnarray}
\end{lem}
\begin{proof}[Proof of Lemma \ref{Lemma_Falk}]
The proof of Lemma \ref{Lemma_Falk} can be found in \cite{Falk}.
\end{proof}

Also, we need the following Lemma which is closely related to the previous one:
\begin{lem}\label{Lemma_an_bn}
	In Lemma \ref{Lemma_Falk} we have
	\begin{equation}\label{Eq_Appendix_Lemma_an_bn}
		\lim_{n \rightarrow \infty} \frac{a_n}{b_n}=\infty,
	\end{equation}
    where $a_n$ and $b_n$ are defined in (\ref{Eq_Appendix_Lemma_Falk_2}).
\end{lem}
\begin{proof}[Proof of Lemma \ref{Lemma_an_bn}]
Proof of Lemma \ref{Lemma_an_bn} can be found in \cite{Pooya_2013}.
\end{proof}
Now, with the help of Lemmas \ref{Lemma_Falk} and \ref{Lemma_an_bn}, we have the following:
\begin{eqnarray}\label{Eq_Appendix_Lemma_Direct_Power_Main}
& & \Pr\{\gamma_{(r),(r)} >  \beta \phi \cap \dots \cap \gamma_{(n),(n)}>  \beta  \phi\}\\ \nonumber
 &\stackrel{(a)}=& \Pr\{\gamma_{(r),(r)} >  \beta \phi\} \\ \nonumber
 &\stackrel{(b)}=& \Pr\{\gamma_{(r),(r)} >  \beta K \mu (m+(1+\delta_2)\zeta_n(n-m)) \} \\ \nonumber
 &\stackrel{(c)}\geq& \Pr\{\gamma_{(r),(r)} >  (1-\delta_3)F^{-1}\left(1-\frac{m}{n}\right) \} \\ \nonumber
 &\stackrel{(d)}=& \Pr\{\gamma_{(r),(r)} >  (1-\delta_3)a_n \} \\ \nonumber
 &=& \Pr\{\frac{\gamma_{(r),(r)} - a_n} {b_n} >  -\delta_3 \frac{a_n}{b_n} \} \\ \nonumber
 &\stackrel{(e)}\rightarrow& 1,
\end{eqnarray}
where $K>1$, (a) is due to (\ref{Eq_Main_Channel_Sorting}), (b) is due to (\ref{Eq_Main_Rigorous_Proof_Definitions}), (c) is due to the required inequality (\ref{Eq_Main_Theorem_Required_Inequality}) in the Theorem \ref{Theorem_Main},
(d) is due to (\ref{Eq_Appendix_Lemma_Falk_2}), and (e) is due to Lemmas \ref{Lemma_Falk} and \ref{Lemma_an_bn}.
\end{proof}
\begin{proof}[Proof of Lemma \ref{Lemma_Interference}]
We have:
\begin{eqnarray}\label{Eq_Appendix_Lemma_Interference_Main}
& &\Pr\{I_r \leq \phi \cap \dots \cap I_n \leq  \phi\}  \\ \nonumber
&=& 1- \Pr\{I_r > \phi \cup \dots \cup I_n > \phi\} \\ \nonumber
&\geq& 1- \sum_{i=r}^{n}{\Pr\{I_i>\phi\}} \\ \nonumber
&=& 1- m\Pr\{I_r>\phi\} \\ \nonumber
&\stackrel{(a)}\rightarrow& 1,
\end{eqnarray}
where (a) can be easily proved by the Large Deviations Principle (LDP) theorems for super-exponential and sub-exponential distributions (see
\cite{Hollander} for LDP theorem for super-exponential and \cite{Mikosch} for LDP theorem for sub-exponential distributions).
\end{proof}
\bibliographystyle{ieeetr}

\end{document}